\documentclass[twocolumn,10pt]{IEEEtran}

\usepackage{epsfig,latexsym}
\usepackage{float}
\usepackage{indentfirst}
\usepackage{amsmath}
\usepackage{amssymb}
\usepackage{times}
\usepackage{subfigure}
\usepackage{psfrag}
\usepackage{cite}
\usepackage{lastpage}
\usepackage{fancyhdr}
\usepackage{color}
 \usepackage{amsthm}
\usepackage{bigints}
\newtheorem{theorem}{Theorem}
\begin{document}%
\title{ {\huge   Impact of User Pairing on 5G Non-Orthogonal Multiple Access }}

\author{ Zhiguo Ding, \IEEEmembership{Member, IEEE},    Pingzhi Fan, \IEEEmembership{Fellow, IEEE},
  and  H. Vincent Poor, \IEEEmembership{Fellow, IEEE}\thanks{
Z. Ding and H. V. Poor  are with the Department of
Electrical Engineering, Princeton University, Princeton, NJ 08544,
USA.   Z. Ding is also with the School of
Computing and Communications, Lancaster
University, LA1 4WA, UK. Pingzhi Fan is with the Institute of Mobile
Communications, Southwest Jiaotong University, Chengdu, China. }\vspace{-0.5em}} \maketitle
\begin{abstract}
 Non-orthogonal multiple access (NOMA) represents  a paradigm shift from conventional orthogonal multiple access (MA) concepts, and has been recognized as one of the key enabling technologies for 5G systems. In this paper, the impact of user pairing on the performance of two NOMA systems, NOMA with fixed power allocation (F-NOMA) and cognitive radio inspired NOMA (CR-NOMA), is characterized.  For F-NOMA, both analytical and numerical results are provided to demonstrate that F-NOMA can offer a larger sum rate than orthogonal MA, and the performance gain of F-NOMA over conventional MA   can be further enlarged by selecting users whose channel conditions are more distinctive. For CR-NOMA, the quality of service (QoS) for   users with the poorer  channel condition can be guaranteed since the transmit power allocated to other users  is constrained following the concept of cognitive radio networks. Because of this constraint, CR-NOMA has   different behavior compared to F-NOMA. For example, for  the user with the best channel condition, CR-NOMA prefers to pair it with the user with the second best channel condition, whereas  the user with the worst   channel condition is preferred by F-NOMA.
\end{abstract}\vspace{-1em}
 \section{Introduction}
 Multiple access in 5G mobile networks is an emerging research topic, since it is   key for the next generation network to keep pace with the exponential growth of mobile data and multimedia traffic \cite{Li5G} and \cite{Huawei5g}. Non-orthogonal multiple access (NOMA) has recently received considerable  attention as a promising candidate for    5G multiple access  \cite{NOMAPIMRC,6933459,Nomading,6708131}. Particularly, NOMA uses the power domain   for    multiple access, where different users are served at different power levels. The users with better channel conditions employ  successive interference cancellation (SIC) to remove the messages intended for  other users before decoding their own \cite{Cover1991}. The benefit of using NOMA can be illustrated by the following example. Consider that there is a user close to the edge of its cell, denoted by $A$, whose channel condition is very poor. For conventional  MA, an  orthogonal bandwidth channel, e.g., a time slot, will be allocated to this user, and the other users cannot use  this time slot.   The key idea of NOMA is to squeeze  another user with better channel condition, denoted by $B$, into this time slot. Since $A$'s channel condition is very poor, the interference from $B$ will not cause much performance degradation to $A$, but the overall system throughput  can be significantly improved since additional  information can be delivered between the base station (BS) and $B$. The design of NOMA for uplink transmissions has been proposed in \cite{6933459}, and the performance of NOMA with randomly deployed mobile stations has been characterized in \cite{Nomading}. The combination of cooperative diversity with NOMA has been considered in \cite{Zhiguo_conoma}.

Since  multiple users are admitted at the same time, frequency and spreading code, co-channel interference will be   strong in NOMA systems, i.e.,  a NOMA system is interference limited. As a result, it may not be realistic to ask all the users in the system to perform  NOMA jointly. A promising alternative is   to build a hybrid MA system, in which    NOMA is combined with conventional MA. In particular,  the users in the system can be divided into multiple groups, where  NOMA is implemented within each group and different groups are allocated with  orthogonal bandwidth resources. Obviously the performance of  this hybrid MA scheme is very  dependent on  which users are   grouped together, and the aim of this paper is to investigate the effect of this grouping.  Particularly, tn this paper, we focus on a downlink communication scenario with one BS and multiple users, where the users are ordered according to their connections to the BS, i.e., the $m$-th user has the $m$-th worst connection to the BS. Consider that two users, the $m$-th user and the $n$-th user,  are selected for performing NOMA jointly, where $m<n$. The impact of user pairing on the performance of NOMA will be characterized in this paper, where two types of NOMA will be considered. One is based on fixed power allocation, termed   F-NOMA, and the other is cognitive radio inspired NOMA, termed   CR-NOMA.

 For the F-NOMA scheme, the probability that F-NOMA can achieve a larger sum rate than conventional MA is first studied, where an exact expression for this probability as well as its high signal-to-noise ratio (SNR) approximation are obtained. These developed analytical results demonstrate that it is almost certain for F-NOMA to outperform conventional MA, and the channel quality of the $n$-th user is critical to this probability.   In addition, the gap between the sum  rates achieved by F-NOMA and conventional MA is also studied, and it is shown  that this gap is determined by how different the two users' channel conditions are, as initially reported in \cite{Zhiguo_conoma}. For example, if $n=M$, it is preferable to choose $m=1$, i.e., pairing the user with the best channel condition with the user with the worst channel condition. The reason for this phenomenon can be explained as follows. When $m$ is small, the $m$-th user's channel condition is poor, and the data rate supported by this user's channel is also small. Therefore the spectral efficiency of conventional MA is low, since the bandwidth   allocated to this user cannot be accessed   by other users. The use of F-NOMA ensures  that  the $n$-th user will  have   access to the resource allocated to the $m$-th user. If $(n-m)$ is small, the $n$-th user's channel quality is similar to the $m$-th user's, and the benefit to use NOMA is limited. But if $n>>m$, the $n$-th user can use the bandwidth resource much  more efficiently than the $m$-th user, i.e.,  a larger $(n-m)$ will result in a larger performance gap between F-NOMA and conventional MA.

 The key idea of CR-NOMA is to opportunistically serve the $n$-th user on the condition that the $m$-th user's quality of service (QoS) is guaranteed. Particularly the transmit power allocated to the $n$-th user is constrained by the $m$-th user's signal-to-interference-noise ratio (SINR), whereas F-NOMA uses a fixed set of power allocation coefficients. Since the $m$-th user's QoS can be guaranteed, we mainly focus on the performance of the $n$-th user offered by CR-NOMA. An  exact expression for the outage probability achieved by CR-NOMA is obtained first, and then  used for the study of the diversity order. In particular, we show that the diversity order experienced by the $n$-th user is $m$, which means that   the $m$-th user's channel quality is critical to the performance of CR-NOMA. This is mainly because of the imposed  SINR constraint, where the $n$-th user can be admitted into the bandwidth channel occupied by the $m$-th user, only if the $m$-th user's SINR is guaranteed. As a result, with a fixed $m$, increasing $n$ does not bring much improvement  to the $n$-th user's outage probability, which is different from F-NOMA. If the ergodic rate is used as the criterion, a similar difference between F-NOMA and CR-NOMA can be observed. Again take the scenario described  in the last paragraph as an example. If $n=M$, in order to yield a large gain over conventional MA, F-NOMA prefers the choice of $m=1$, but CR-NOMA prefers the choice of $m=M-1$ , i.e., pairing the user with the best channel condition with the user with the second best channel condition.

 \section{NOMA With Fixed Power Allocation  }\label{section noma}
Consider a downlink communication scenario with one BS and $M$ mobile users.   Without loss of generality, assume that the users' channels have been ordered as $|h_1|^2\leq \cdots \leq |h_M|^2$, where $h_m$ denotes the Rayleigh fading channel gain between the BS and the ordered $m$-th user.  Consider that  the $m$-th user and the $n$-th user, $m<n$, are paired to perform NOMA.

In this section, we focus on F-NOMA, where the BS allocates a fixed amount of transmit power to each user. In particular, denote $a_m$ and $a_n$ as the  power allocation coefficients for the two users, where these coefficients are fixed  and $a_m^2+a^2_n=1$.  According to the principle of NOMA, $a_m\geq a_n$ since $|h_m|^2\leq |h_n|^2$.    The rates achievable to the two users are given by
\begin{eqnarray}
R_m = \log \left( 1+\frac{|h_m|^2a_m^2}{|h_m|^2a_n^2+\frac{1}{\rho}} \right),
\end{eqnarray}
and
\begin{eqnarray}
R_n = \log \left( 1+ \rho a^2_n |h_n|^2\right),
\end{eqnarray}
respectively, where $\rho$ denotes the transmit SNR. Note that the $n$-th user can decode the message intended for the $m$-th user successfully  and  $R_n$ is always achievable at the $n$-th user,
since $R_m\leq \log \left( 1+\frac{|h_n|^2a_m^2}{|h_n|^2a_n^2+\frac{1}{\rho}} \right)$.

On the other hand, an orthogonal MA scheme, such as time-division multiple-access (TDMA), can support the following data rate:
\begin{eqnarray}
\bar{R}_i = \frac{1}{2}\log \left( 1+ \rho |h_i|^2\right),
\end{eqnarray}
where $i\in\{m,n\}$.  In the following subsections, the impact of user pairing on the sum rate and the individual user rates achieved by F-NOMA is investigated.

\subsection{Impact of user pairing on the sum rate}
In this subsection, we focus on how user pairing affects the probability that NOMA  achieves a lower   sum rate than conventional MA schemes, which is given by
\begin{eqnarray}
\mathrm{P}(R_m+R_n<\bar{R}_m+\bar{R}_n).
\end{eqnarray}
The following theorem provides an exact expression for the above probability as well as its high SNR approximation.

\begin{theorem}\label{lemma1}
Suppose  that the $m$-th and $n$-th ordered users are paired to perform NOMA. The probability that F-NOMA   achieves a lower  sum rate than conventional MA is given by
\begin{align}
&\mathrm{P}(R_m+R_n<\bar{R}_m+\bar{R}_n)=\\\nonumber &1- \sum^{n-1-m}_{i=0}{n-1-m \choose i} \frac{(-1)^i\varpi_1}{m+i} \int^{\varpi_2}_{\varpi_4} f(y)(F(y))^{n-1-m-i} \\\nonumber &\times    \left(1-F(y)\right)^{M-n}  \left(\left[F\left(y\right)\right]^{m+i}-\left[F\left(\frac{\varpi_2-y}{1+y}\right)\right]^{m+i}\right)dy\\ \nonumber &-\frac{\varpi_3}{\rho}\sum^{n-1}_{j=0}{n-1 \choose j}(-1)^j \frac{\rho}{M-n+j+1} e^{-\frac{(M-n+j+1)\varpi_2}{\rho}},
\end{align}
where $f(x)=\frac{1}{\rho}e^{-\frac{x}{\rho}}$, $F(x) = 1-e^{-\frac{x}{\rho}}$,
$\varpi_1=\frac{M!}{(m-1)!(n-1-m)!(M-n)!}$, $\varpi_2=\frac{1-2a_n^2}{a_n^4}$, $\varpi_3=\frac{M!}{(n-1)!(M-n)!}$ and $\varpi_4=\sqrt{1+\varpi_2}-1$.
At high SNR, this probability can be approximated as follows:
\begin{align}
&\mathrm{P}(R_m+R_n<\bar{R}_m+\bar{R}_n) \approx\frac{1}{\rho^n}\left(
   \frac{ \varpi_3\varpi_2^n}{n} - \varpi_1\varpi\right),
\end{align}
where $\varpi=\sum^{n-1-m}_{i=0}{n-1-m \choose i} \frac{(-1)^i}{m+i} \int^{\varpi_2}_{\varpi_4} y^{n-1-m-i}    $ $\times    \left( y^{m+i}-\left[\frac{\varpi_2-y}{(1+y)}\right]^{m+i}\right)dy$, i.e., $\varpi$ is a constant and not a function of $\rho$.
\end{theorem}
\begin{proof}
See the appendix.
\end{proof}
Theorem \ref{lemma1} demonstrates that it is almost certain for F-NOMA to outperform conventional MA, particularly at high SNR.    Furthermore, the  decay rate of the probability $\mathrm{P}(R_m+R_n<\bar{R}_m+\bar{R}_n)$ is approximately $\frac{1}{\rho^n}$, i.e., the quality of the $n$-th user's channel determines the decay rate of this  probability.

\subsection{Asymptotic studies of the sum rate achieved by NOMA}
In addition to the probability   $\mathrm{P}(R_m+R_n<\bar{R}_m+\bar{R}_n)$, it is also of interest  to study how large  of a performance gain F-NOMA offers over conventional MA, i.e.,
\[
\mathrm{P}(R_m+R_n-\bar{R}_m-\bar{R}_n<R),
\]
where $R$ is a targeted performance gain. The probability studied in the previous subsection can be viewed as a special case by setting $R=0$. An interesting observation for the cases with $R>0$ is that there will be an error floor for $\mathrm{P}(R_m+R_n-\bar{R}_m-\bar{R}_n<R)$, regardless of how large the SNR is. This can be shown  by studying the following asymptotic expression of the sum rate gap:
\begin{align}
&R_m+R_n-\bar{R}_m-\bar{R}_n\\ \nonumber \underset{\rho\rightarrow \infty}{\rightarrow} &
 \log \left( \frac{1}{a_n^2} \right)+ \log \left(  \rho a^2_n |h_n|^2\right) -  \log \left( \rho|h_m|         |h_n|\right)
 \\ \nonumber =&    \log       |h_n|  -  \log    |h_m|      ,
\end{align}
which is not a function of SNR.
Hence the probability can be  expressed asymptotically as follows:
\begin{align}\label{asym}
&\mathrm{P}\left(R_m+R_n-\bar{R}_m-\bar{R}_n<R\right)\\ \nonumber \underset{\rho\rightarrow \infty}{\rightarrow}  &\mathrm{P}\left(   \log       |h_n|  -  \log    |h_m|    <R \right) .
\end{align}
When $R=0$, $\mathrm{P}\left(R_m+R_n-\bar{R}_m-\bar{R}_n<R\right)\rightarrow 0$, which is consistent with Theorem  \ref{lemma1}, since $$\mathrm{P}\left(R_m+R_n<\bar{R}_m+\bar{R}_n\right)\sim \frac{1}{\rho^n}\underset{\rho\rightarrow \infty}{\rightarrow}0.$$

When $R\neq 0$, \eqref{asym} implies that the probability $\mathrm{P}\left(R_m+R_n-\bar{R}_m-\bar{R}_n<R\right)$ can be  expressed  asymptotically as follows:
\begin{align}
 \mathrm{P}\left(   \log       |h_n|  -  \log    |h_m|    <R \right)   \rightarrow
 \mathrm{P}\left(          \frac{|h_n|^2}{    |h_m|^2}    <2^{2R} \right).
\end{align}
Directly  applying the joint probability density function (pdf) of the users' channels shown in \eqref{joint pdf}, the probability can be rewritten as   follows:
\begin{align}
 &\mathrm{P}\left(   \log       |h_n|  -  \log    |h_m|    <R \right) \\ \nonumber & =
 \int^{\infty}_{0}\int^{y}_{2^{-2R}y}\varpi_1 f(x) f(y)[F(x)]^{m-1} \\\nonumber &\times \left(F(y)-F(x)\right)^{n-1-m} \left(1-F(y)\right)^{M-n}dx dy ,
\end{align}
which is quite complicated to evaluate.
%Again applying binomial expansion, we have the following
%\begin{align}
% &\mathrm{P}\left(   \log       |h_n|  -  \log    |h_m|    <R \right) \\ \nonumber & =
% \sum^{n-1-m}_{i=0}{n-1-m \choose i} (-1)^i \int^{\infty}_{0}\int^{y}_{2^{-2R}y}\varpi_1 f(x) f(y)[F(x)]^{m-1+i} \\\nonumber &\times \left(F(y)\right)^{n-1-m-i} \left(1-F(y)\right)^{M-n}dx dy
%  \\ \nonumber & =
% \sum^{n-1-m}_{i=0}{n-1-m \choose i} (-1)^i \int^{\infty}_{0} \varpi_1 f(y)\frac{1}{m+i}\left([F(y)]^{m-1+i} -[F(2^{-2R}y)]^{m-1+i}\right) \\\nonumber &\times \left(F(y)\right)^{n-1-m-i} \left(1-F(y)\right)^{M-n}dx dy .
%\end{align}
 In \cite{Subranhmaniam}, a simpler pdf for the ratio of two order statistics  has been provided  as follows:
 \begin{align}\nonumber
 f_{ \frac{|h_m|^2}{    |h_n|^2} }(z) = \frac{M!}{(m-1)!(n-m-1)!(M-n)!} \sum^{m-1}_{j_1=0}\sum^{n-m-1}_{j_2=0}\\ \nonumber  (-1)^{j_1+j_2} {m-1 \choose j_1} {n-m-1 \choose j_2} (\tau_2+\tau_1 z)^{-2},
 \end{align}
 where $\tau_1=j_1-j_2+n-m$ and $\tau_2=M-n+1+j_2$. By using this pdf,  the addressed probability can be calculated as follows:
 \begin{align}\label{sum rate gap}
 &\mathrm{P}\left(   \log       |h_n|  -  \log    |h_m|    <R \right) \\ \nonumber & \rightarrow
  \frac{M!}{(m-1)!(n-m-1)!(M-n)!} \sum^{m-1}_{j_1=0}\sum^{n-m-1}_{j_2=0}\\ \nonumber  & \frac{(-1)^{j_1+j_2}}{\tau_1} {m-1 \choose j_1} {n-m-1 \choose j_2} \left(\frac{1}{\tau_2+2^{-2R}\tau_1} - \frac{1}{\tau_2+\tau_1} \right) .
\end{align}

\subsection{Impact of user pairing on individual user rates}
Careful user pairing  not only improves the sum rate, but also has the  potential to improve the individual user rates, as shown in this section.    We first focus on the probability that F-NOMA can achieve a  larger rate  than orthogonal MA for the $m$-th user which is given by
\begin{align}
&\mathrm{P}(R_m>\bar{R}_m)  \\ \nonumber =&   \mathrm{P}\left(  \left( 1+\frac{|h_m|^2a_m^2}{|h_m|^2a_n^2+\frac{1}{\rho}} \right)^2>   (1+\rho|h_m|^2)\right).
\end{align}
After some algebraic manipulations, the above probability can be further rewritten  as follows:
\begin{align}\label{user m}
&\mathrm{P}(R_m>\bar{R}_m) = \mathrm{P}\left(  |h_m|^2<\frac{1-2a_n^2}{\rho a_n^4}\right)\\ \nonumber &=  \int^{\frac{1-2a_n^2}{\rho a_n^4}}_{0}  \frac{\varpi_5 }{\rho}e^{-\frac{(M-m+1)y}{\rho}} \left(1- e^{-\frac{y}{\rho}}\right)^{m-1}  dy
\\ \nonumber &=  \sum^{m-1}_{i=0} {m-1 \choose i} \frac{(-1)^i\varpi_5}{M-m+i+1}\left(1-e^{-\frac{(1-2a_n^2)(M-m+i+1)}{\rho a_n^4}}\right)  , \end{align}
where $\varpi_5=\frac{M!}{(m-1)!(M-m)!}$.

By applying a series  expansion, the above probability can be rewritten as follows:
\begin{align}\label{indvidtual 1}
\mathrm{P}(R_m>\bar{R}_m)
 &=  \sum^{m-1}_{i=0} {m-1 \choose i}  (-1)^{i+1}\varpi_5  \\ \nonumber&\times \sum^{\infty}_{k=1}(-1)^k \frac{(1-2a_n^2)^{k}(M-m+i+1)^{k-1}}{k!\rho^{k} a_n^{4k}}.
  \end{align}
Again applying the results in \eqref{property1} and \eqref{property2}, the above equation can be approximated as follows:
\begin{align}\label{user m high}
\mathrm{P}(R_m>\bar{R}_m)
 &\approx   \varpi_5   \frac{(1-2a_n^2)^{m}}{m\rho^{m} a_n^{4m}},
  \end{align}
  which means that $\mathrm{P}(R_m>\bar{R}_m)$ decays at a rate of $\frac{1}{\rho^m}$.

On the other hand, the probability that the $n$-th user can experience better performance in a NOMA system than in orthogonal MA systems  is given by
\begin{eqnarray}\nonumber
\mathrm{P}(R_n>\bar{R}_n) = \mathrm{P}\left(\log \left( 1+ \rho a^2_n |h_n|^2\right)> \frac{1}{2}\log (1+\rho|h_n|^2 \right).
\end{eqnarray}
Following similar steps as previously, we obtain the following:
\begin{eqnarray} \label{user n}
\mathrm{P}(R_n>\bar{R}_n) = \mathrm{P}\left( |h_n|^2>\frac{1-2a_n^2}{\rho a_n^4}\right).
\end{eqnarray}
Interestingly $\mathrm{P}(R_n>\bar{R}_n)$ in \eqref{user n} is very much similar to $\mathrm{P}(R_m>\bar{R}_m)$ in \eqref{user m}, which yields the following:
\begin{align}\label{individutal 2}
&\mathrm{P}(R_n>\bar{R}_n) = 1-   \sum^{n-1}_{i=0} {n-1 \choose i} \frac{(-1)^i\varpi_3}{M-n+i+1}\\ \nonumber &\times \left(1-e^{-\frac{(1-2a_n^2)(M-n+i+1)}{\rho a_n^4}}\right),  \end{align}
and its high SNR approximation is given by
\begin{align}\label{user n high}
\mathrm{P}(R_n>\bar{R}_n)
 &\approx   1- \varpi_3   \frac{(1-2a_n^2)^{n}}{n\rho^{n} a_n^{4n}}.
  \end{align}
  As can be seen from \eqref{user m high} and \eqref{user n high}, the two users will have totally different experience in NOMA systems. Particularly, a user with a better channel condition is more willing to perform NOMA since   $\mathrm{P}(R_n>\bar{R}_n)\rightarrow 1$, which is not true for a user with a poor channel condition. Furthermore, it is preferable to pair two users whose channel conditions are significantly distinct, since \eqref{user m high} and \eqref{user n high} implies that $m$ should be as small as possible and $n$ should be as large as possible.

\begin{figure*}[bt!]
\begin{align}\label{theorem1}
&\mathrm{P}_{n}^o=\varpi_5  \sum^{M-n}_{i=0}{M-n \choose i}(-1)^i  \frac{\left[G(b)\right]^{m+i}}{m+i} + \sum^{n-1-m}_{i=0}{n-1-m \choose i} (-1)^i\int^{a\epsilon_1}_b g(y) \left(1-G(y)\right)^{M-n}G(y)^{n-1-m-i} \varpi_1 \\ \nonumber &\times \frac{\left(G(y)^{m+i} - G(b)^{m+i}\right)}{m+i}dy +\sum^{n-1-m}_{i=0}{n-1-m \choose i} (-1)^i\int^{b+a\epsilon_1}_{a\epsilon_1} \left(1-G(y)\right)^{M-n}G(y)^{n-1-m-i}\frac{\left(G(y)^{m+i} - G(b)^{m+i}\right)}{m+i}\\\nonumber &\times \varpi_1 g(y) dy +\sum^{n-1-m}_{i=0}{n-1-m \choose i} (-1)^i\int^\infty_{b+a\epsilon_1} g(y) \left(1-G(y)\right)^{M-n} G(y)^{n-1-m-i} \varpi_1  \frac{\left(G\left(\frac{b}{1-\frac{a\epsilon_1}{|h_n|^2}}\right)^{m+i} - G(b)^{m+i}\right)}{m+i}dy.
\end{align}
\end{figure*}
\section{ Cognitive Radio Inspired  NOMA}
NOMA can  be also  viewed as a special case of cognitive radio systems \cite{4840529} and \cite{5403537}, in which a user with a strong channel condition, viewed as a secondary user, is squeezed into  the spectrum occupied  by  a user with a poor channel condition, viewed as a primary user. Following the concept of cognitive radio networks, a variation of NOMA, termed as CR-NOMA, can be designed as follows. Suppose  that  the BS needs to serve  the $m$-th user, i.e., a user a with poor channel condition, due to either the high priority of this user's messages or user fairness, e.g., this user has not been served for a long time.  This user can be viewed as a primary user in a cognitive radio system. The $n$-th user can   be admitted into this   channel on the condition that the $n$-th user will not cause too much performance degradation to the $m$-th user.

Consider that the targeted SINR at the $m$-th user is $I$, which means that the choices of the power allocation coefficients, $a_m$ and $a_n$,  need to satisfy the following constraint:
\begin{align}\label{an dymanic2}
\frac{|h_m|^2a_m^2}{|h_m|^2a_n^2+\frac{1}{\rho}} \geq I.
\end{align}
This means that the maximal transmit power that  can be allocated to the $n$-th user is given by
\begin{align}\label{an dymanic1}
a_n^2= \max \left\{0,\frac{|h_m|^2 -\frac{I}{\rho}}{|h_m|^2(1+I)}\right\},
\end{align}
which means that $a_n=0$ if $|h_m|^2 <\frac{I}{\rho}$. 
Note that the choice of $a_n$ in \eqref{an dymanic1} is a function of the channel coefficient $h_m$, unlike  the constant choice of $a_n$ used by F-NOMA in the previous section.

Since the $m$-th user's QoS can be guaranteed due to \eqref{an dymanic2}, we only need to study  the performance experienced  by the $n$-th user. Particularly  the outage performance of the $n$-th user is defined as follows:
\begin{align}
\mathrm{P}_o^n\triangleq \mathrm{P}\left(\log(1+a_n^2\rho|h_n|^2)<R\right),
\end{align}
and the following theorem provides an exact expression for the above outage probability as well as its approximation.
\begin{theorem}\label{theorm2}
Suppose  that the transmit       power allocated to the $n$-th user can satisfy the predetermined SINR threshold, $I$, as shown in \eqref{an dymanic1}.  The $n$-th user's  outage probability achieved by CR-NOMA   is given by  \eqref{theorem1}, where $g(y)=e^{-y}$, $G(y) = 1-e^{-y}$, $\epsilon_1=\frac{2^R-1}{\rho}$, $b=\frac{I}{\rho}$,  $a=1+I$ and $b\leq a\epsilon_1$. The diversity order achieved by CR-NOMA is given by
\[
\underset{\rho \rightarrow \infty}{\lim}- \frac{\log P_o^n}{\log \rho} = m.
\]

\end{theorem}
\begin{proof}See the appendix.
\end{proof}

Theorem \ref{theorm2} demonstrates an interesting phenomenon that, in CR-NOMA, the diversity order experienced by the $n$-th user is determined by how good the $m$-th user's channel quality is. This is because the $n$-th user can be admitted to the channel occupied by the $m$-th user only  if the $m$-th user's QoS is met. For example,   if the $m$-th user's channel is poor and its targeted SINR is   high, it is very likely that the BS allocates all the power to the $m$-th user, and the $n$-th user might not even get served.

Recall from the previous section that   F-NOMA can achieve a diversity gain of $n$ for the $n$-th user, and therefore the diversity order achieved by CR-NOMA could be much smaller than that  achieved by F-NOMA, particularly if $n>>m$. This performance difference is again  due to  the imposed power constraint shown in \eqref{an dymanic1}.

  It is important to point out that CR-NOMA can strictly guarantee the $m$-th user's QoS, and therefore achieve better fairness compared to F-NOMA. In particular,  the use of   CR-NOMA can ensure that a diversity order of $m$ is achievable to the $n$-th user, and     admitting the $n$-th user into the same channel as the $m$-th user will not cause too much performance degradation to the $m$-th user. Particularly  the SINR experienced  by the $m$-th user is strictly maintained at the predetermined level $I$.

 \subsection*{Sum rate achieved by CR-NOMA}
Without sharing the spectrum with the $n$-th user, i.e, all the bandwidth resource is allocated to the $m$-th user, the following rate is achievable:
\begin{eqnarray}
\tilde{R}_m =  \log \left( 1+ \rho |h_m|^2\right).
\end{eqnarray}
It is easy to show that the use of CR-NOMA always achieves a larger sum  rate since
\begin{align}
&R_m+R_n-\tilde{R}_m\\\nonumber
=& \log \left( 1+\frac{|h_m|^2a_m^2}{|h_m|^2a_n^2+\frac{1}{\rho}} \right)+\log \left( 1+ \rho a^2_n |h_n|^2\right) \\ \nonumber &-  \log \left(1+\rho |h_m|^2\right)
\\\nonumber
=&\log \frac{ 1+ \rho a^2_n |h_n|^2}{ 1+ \rho a^2_n |h_m|^2}\geq 0.
\end{align}
 This  superior performance gain is not surprising, since the key idea of CR-NOMA is to serve a user with a strong channel condition, without causing too much performance degradation to the user with a poor channel condition.

 In addition,  it is   of interest to study how much the averaged rate gain CR-NOMA can yield, i.e., $\mathcal{E}\left\{R_n  \right\}$. This averaged rate gain can be calculated as follows:
\begin{align}
\mathcal{E}\left\{R_n  \right\} &= \int^{\infty}_{b} \int^{\infty}_{x}\log\left(1+\frac{x -b}{xa}\rho y\right)\\ \nonumber &\times f_{|h_{m}|^2,|h_{n}|^2}(x,y) dydx.
\end{align}
In general, the evaluation of the above equation is difficult, and in the following we provide a case study when  $n-m=1$. Particularly, the joint pdf of the channels for this special case can be simplified and the averaged rate gain can calculated as follows:
\begin{align}\nonumber
\mathcal{E}\left\{R_n  \right\} &= \varpi_1\int^{\infty}_{b}f(x) [F(x)]^{m-1} \int^{\infty}_{x}\log\left(1+\frac{x -b}{xa}\rho y\right)\\   &\times  f(y)  \left(1-F(y)\right)^{M-n}dydx\\ \nonumber
 &= \frac{-\varpi_1}{M-n+1}\int^{\infty}_{b}f(x) [F(x)]^{m-1} \\   \nonumber &\times\int^{\infty}_{x}\log\left(1+\frac{x -b}{xa}\rho y\right) d\left(1-F(y)\right)^{M-n+1}dx.
\end{align}
After some algebraic manipulations, the above equation can be rewritten as follows:
\begin{align}\nonumber
\mathcal{E}\left\{R_n  \right\} &=  \frac{\varpi_1}{M-n+1}\int^{\infty}_{b}f(x) [F(x)]^{m-1} \\   \nonumber &\times\left( \log\left(1+\frac{x -b}{a}\rho \right) \left(1-F(x)\right)^{M-n+1}\right. \\ \nonumber &\left.+\frac{1}{\ln 2}\int^{\infty}_{x}\left(1-F(y)\right)^{M-n+1} \frac{\frac{x -b}{xa}\rho }{ 1+\frac{x -b}{xa}\rho y} dy\right)dx.
\end{align}
Now applying Eq. (3.352.2) in \cite{GRADSHTEYN}, the average rate gain can be expressed as follows:
\begin{align}\label{cr noma rate}
\mathcal{E}\left\{R_n  \right\} &=  \frac{\varpi_1}{M-n+1}\int^{\infty}_{b}f(x) [F(x)]^{m-1} \\   \nonumber &\times\left( \log\left(1+\frac{x -b}{a}\rho \right) \left(1-F(x)\right)^{M-n+1}- \frac{e^{\frac{x^2a}{\rho(x -b)}}}{\ln 2}\right. \\ \nonumber &\left.\times\text{Ei}\left(-(M-n+1)x-\frac{(M-n+1)xa}{\rho(x -b)}\right)\right)dx,
\end{align}
where $\text{Ei}(\cdot)$ denotes the exponential integral.

\begin{figure}[!htp]
\begin{center} \subfigure[ $m=1$ ]{\includegraphics[width=0.47\textwidth]{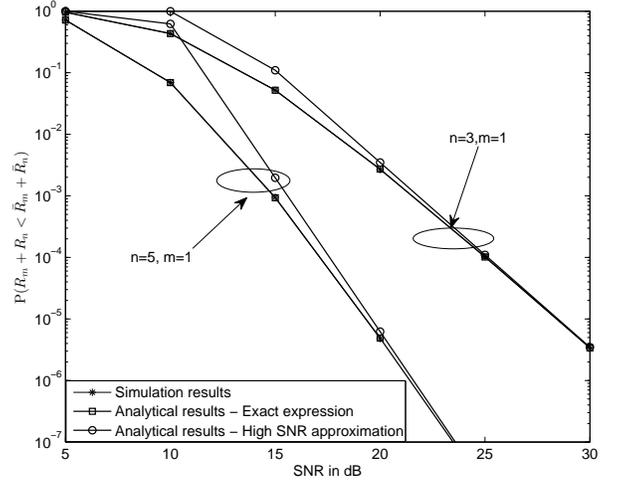}}
\subfigure[$m=2$]{\includegraphics[width=0.47\textwidth]{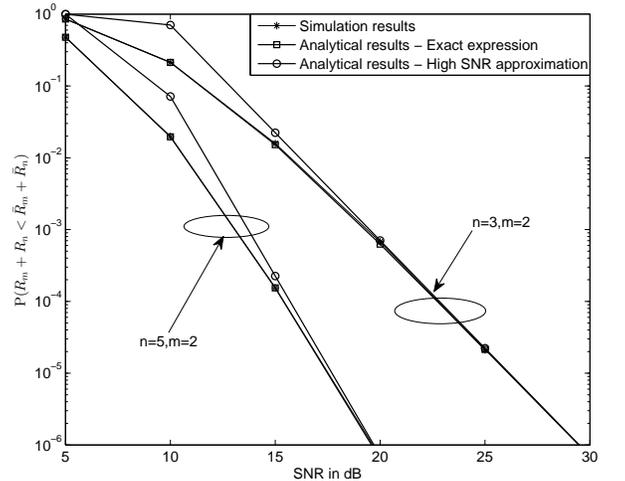}}
\end{center}
 \caption{The probability that F-NOMA realizes  a lower sum rate than conventional MA. $M=5$. The analytical results are based on Theorem \ref{lemma1}.  }\label{fig 1}
\end{figure}

\section{Numerical Studies}
In this section, computer simulations are used to evaluate the performance of two NOMA schemes as well as the accuracy of the developed analytical results.

 \subsection{NOMA with fixed power allocation}
 In Fig. \ref{fig 1}, the probability that F-NOMA realizes a lower sum rate than conventional MA, i.e., $\mathrm{P}(R_m+R_n<\bar{R}_m+\bar{R}_n)$, is shown as a function of SNR. $a_m^2=\frac{4}{5}$ and $a_n^2=\frac{1}{5}$. As can be seen from both figures,    F-NOMA almost always outperforms conventional MA, particularly at high SNR. The simulation results in Fig. \ref{fig 1} also demonstrate the accuracy of the analytical results provided in Theorem \ref{lemma1}. For example, the exact expression of $\mathrm{P}(R_m+R_n<\bar{R}_m+\bar{R}_n)$ shown in Theorem \ref{lemma1} matches  perfectly with the simulation results, whereas the developed approximation results become   accurate at high SNR.

Another important observation from Fig. \ref{fig 1} is that increasing $n$, i.e., scheduling a user with a better channel condition, will make the probability decrease at a faster rate. This observation is consistent to the high SNR  approximation results provided  in Theorem  \ref{lemma1} which show  that  the slope of the curve for the probability  $\mathrm{P}(R_m+R_n<\bar{R}_m+\bar{R}_n)$ is a function of $n$. In Fig. \ref{fig2}, the probability $\mathrm{P}(R_m+R_n-\bar{R}_m-\bar{R}_n<R)$ is shown with different choices of $R$. Comparing Fig. \ref{fig 1} to Fig. \ref{fig2}, one can observe that $\mathrm{P}(R_m+R_n-\bar{R}_m-\bar{R}_n<R)$ never approaches   zero, regardless of how large the SNR is. This observation confirms the analytical results developed in \eqref{sum rate gap} which show that the probability $\mathrm{P}(R_m+R_n-\bar{R}_m-\bar{R}_n<R)$ is no longer a function of SNR, when $\rho\rightarrow 0$. It is interesting to observe that the choice of a smaller $m$ is preferable to reduce $\mathrm{P}(R_m+R_n-\bar{R}_m-\bar{R}_n<R)$, a phenomenon previously reported in \cite{Zhiguo_conoma}.

\begin{figure}[!htbp]\centering
    \epsfig{file=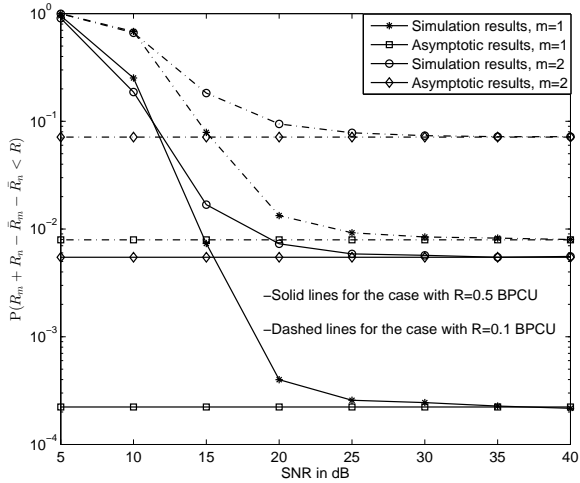, width=0.47\textwidth, clip=}
\caption{ The probability that the sum rate gap between F-NOMA and conventional MA is larger than $R$. $M=5$ and $n=M$. The analytical results are based on \eqref{sum rate gap}.  }\label{fig2}
\end{figure}

In Fig. \ref{fig3}, two different but related probabilities are shown together. One is $\mathrm{P}(R_m>\bar{R}_m)$, i.e., the probability that it is beneficial for the user with a poor channel condition to perform F-NOMA,  and the other is $\mathrm{P}(R_n<\bar{R}_n)$, i.e., the probability that the user with a strong channel condition  prefers conventional MA. In Section \ref{section noma}.C, analytical results have been developed to show that both $\mathrm{P}(R_m>\bar{R}_m)$ and $\mathrm{P}(R_n<\bar{R}_n)$ are decreasing with  increasing SNR, which is confirmed by the simulation results in Fig. \ref{fig3}. The reason that  $\mathrm{P}(R_m>\bar{R}_m)$ is reduced at a higher SNR is that the $m$-th user's rate in an F-NOMA system becomes a constant, i.e., $\log \left( 1+\frac{|h_m|^2a_m^2}{|h_m|^2a_n^2+\frac{1}{\rho}} \right)\underset{\rho\rightarrow \infty}{\rightarrow} \log \left( 1+\frac{ a_m^2}{a_n^2} \right)$, which is much smaller than $\bar{R}_m$, at high SNR. On the other hand,  it is more likely for $R_n$ to be  larger than $\bar{R}_n$ since there is a factor of $\frac{1}{2}$ outside of the logarithm of $\bar{R}_n$.

\begin{figure}[!htbp]\centering
    \epsfig{file=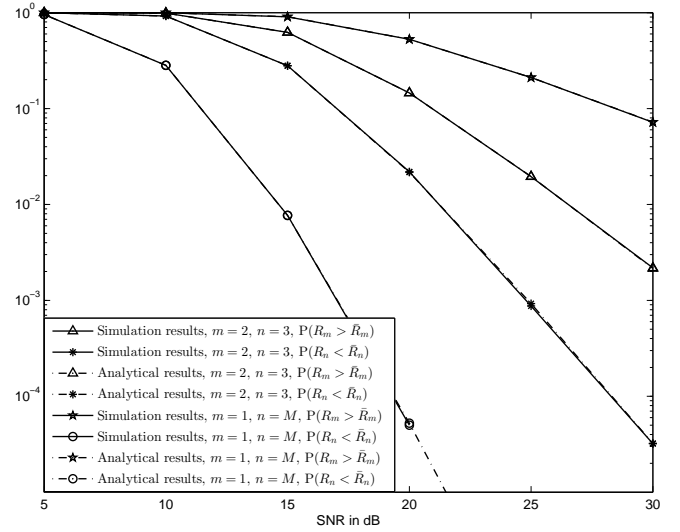, width=0.47\textwidth, clip=}
\caption{ The behavior of individual data rates achieved by F-NOMA, $\mathrm{P}(R_n<\bar{R}_n)$ and $\mathrm{P}(R_m>\bar{R}_m)$. $M=5$.  The analytical results are based on \eqref{indvidtual 1} and \eqref{individutal 2}.  }\label{fig3}
\end{figure}

\subsection{Cognitive radio inspired NOMA}
In Fig. \ref{fig4} the $n$-th user's outage probability achieved by CR-NOMA is shown as a function of SNR. As can be seen from the figure, the exact  expression for the outage probability $\mathrm{P}_n^o\triangleq \mathrm{P}(R_n<R)$ developed in Theorem \ref{theorm2} matches     the simulation results perfectly. Recall from Theorem~\ref{theorm2} that the diversity order achievable for  the $n$-th user is $m$. Or in other words, the slope of the outage probability is determined by the channel quality of the $m$-th user, which is also confirmed by Fig. \ref{fig4}. For example, when increasing $m$ from $1$ to $2$, the outage probability is significantly reduced, and its slope is also increased. To clearly demonstrate the diversity order,   we have provided an auxiliary curve in the figure which  is proportional to $\frac{1}{\rho^m}$.  As can be observed in the figure, this auxiliary curve is parallel to the one for  $ \mathrm{P}(R_n<R)$, which confirms that the diversity order achieved by CR-NOMA is $m$.

\begin{figure}[!htbp]\centering
    \epsfig{file=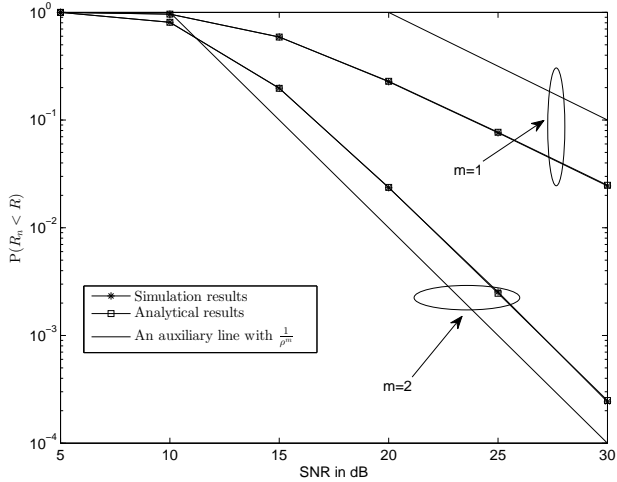, width=0.47\textwidth, clip=}
\caption{ The outage probability for the $n$-th user achieved by CR-NOMA, when $n=M$. $M=5$, $R=1$ bit per channel use (BPCU) and $I=5$. The analytical results are Theorem \ref{theorm2}.  }\label{fig4}
\end{figure}

\begin{figure}[!htbp]\centering
    \epsfig{file=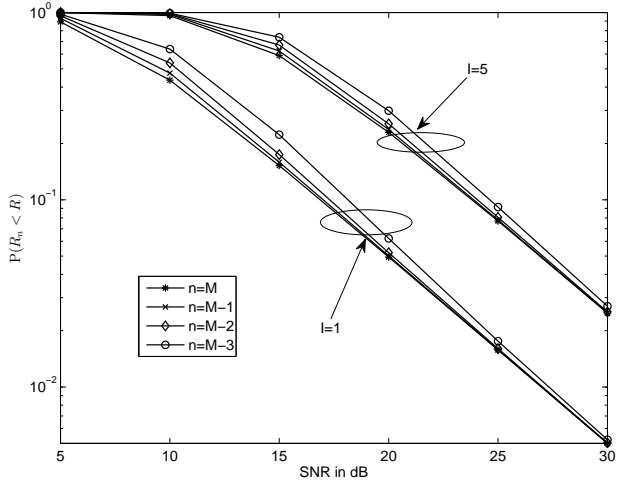, width=0.47\textwidth, clip=}
\caption{ The outage probability for the $n$-th user achieved by CR-NOMA. $m=1$,  $M=5$, and $R=1$ BPCU. }\label{fig5}
\end{figure}
Since Theorem \ref{theorm2} states that the diversity order of $ \mathrm{P}(R_n<R)$ is not a function of $n$, an interesting question is whether a different choice of $n$ matters. Fig. \ref{fig5} is provided to answer this question. While the use of a larger $n$ does bring some reduction of $ \mathrm{P}(R_n<R)$, the performance gain of increasing $n$ is negligible, particularly at high SNR. This is because the channel quality of the $m$-th user becomes a bottleneck for admitting the $n$-th user into the same   channel.

\begin{figure}[!htp]
\begin{center} \subfigure[ $n=m+1$ ]{\includegraphics[width=0.47\textwidth]{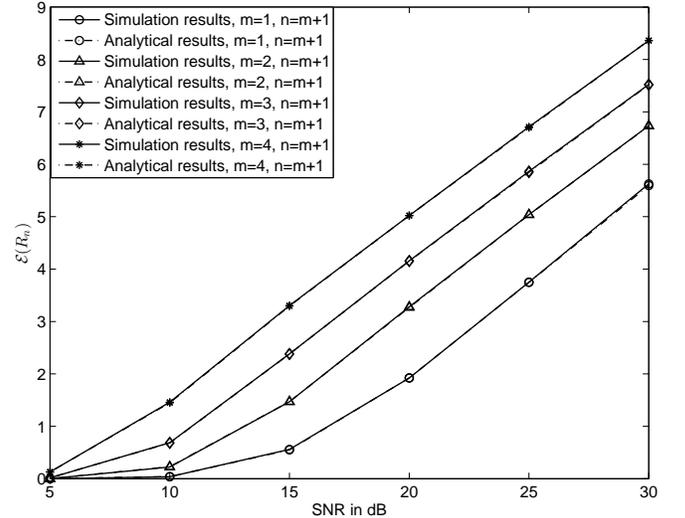}}
\subfigure[General Cases]{\label{fig set comparison
b2}\includegraphics[width=0.47\textwidth]{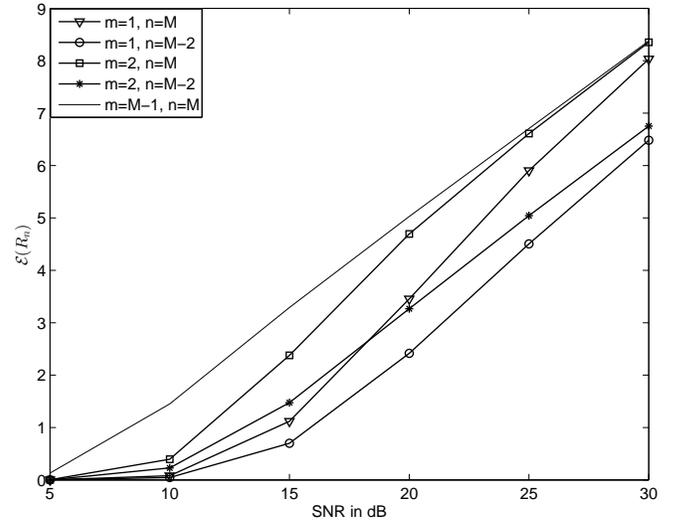}}
\end{center}
  \caption{The ergodic data rate for the $n$-th user achieved by CR-NOMA. $M=5$ and $I=5$. Analytical results are based on \eqref{cr noma rate}.  }\label{fig6}
\end{figure}

In Fig. \ref{fig6} the performance of CR-NOMA is evaluated by using the ergodic data rate as the criterion. Due to the use of \eqref{an dymanic1}, the $m$-th user's QoS can be satisfied, and therefore we only focus on the $n$-th user's data rate, which is the performance gain of CR-NOMA over conventional MA. Fig. \ref{fig6} demonstrates that, by fixing $(n-m)$, it is beneficial to select two users with better channel conditions. While Fig. \ref{fig5} shows that changing  $n$ with a fixed $m$ does not affect the outage probability, Fig. \ref{fig6} demonstrates that user pairing has a significant impact on the ergodic rate. Specifically, when fixing the choice of $m$, pairing it with a user with a better channel condition can yield a gain of more than  $1$ bit per channel use (BPCU)  at $30$dB. Another interesting observation from Fig. \ref{fig6} is that with a fixed $n$, increasing $m$ will improve the performance of CR-NOMA, which is  different from F-NOMA. For example, when $n=M$, Fig. \ref{fig2} shows that the user with the worst channel condition, $m=1$, is the best partner, whereas Fig. \ref{fig6} shows that the user with the second best channel condition, i.e., $m=M-1$, is the best choice.

\section{Conclusions}
In this paper  the impact of user pairing on the performance of two NOMA systems, NOMA with fixed power allocation (F-NOMA) and cognitive radio inspired NOMA (CR-NOMA), has been studied.  For F-NOMA, both analytical and numerical results have been  provided to demonstrate that F-NOMA can offer a larger sum rate than orthogonal MA, and the performance gain of F-NOMA over conventional MA   can be further enlarged by selecting users whose channel conditions are more distinctive. For CR-NOMA, the channel quality of the user with a poor channel condition is critical, since  the transmit power allocated to the other user is constrained following the concept of cognitive radio networks.  One promising future direction of this paper is that the analytical results can be used as criteria  designing  distributed approaches for  dynamic user pairing/grouping.
 \bibliographystyle{IEEEtran}
\bibliography{IEEEfull,trasfer}

\appendix
\textsl{Proof for Theorem \ref{lemma1}:}
Observe that the sum rate achieved by NOMA can be expressed  as follows:
\begin{align}\nonumber
R_m+R_n  = \log \left( \frac{1+\rho|h_m|^2}{\rho|h_m|^2a_n^2+1} \right)  \left( 1+ \rho a^2_n |h_n|^2\right).\end{align}
On the other hand the sum rate achieved by conventional MA is given by
\begin{align}
\bar{R}_m+\bar{R}_n  = \log \left( 1+\rho|h_m|^2  \right)^{\frac{1}{2}}  \left( 1+ \rho   |h_n|^2\right)^{\frac{1}{2}}.\end{align}
Now the addressed probability can be written as follows:
\begin{align}
&\mathrm{P}(R_m+R_n>\bar{R}_m+\bar{R}_n)\\\nonumber &=\mathrm{P}\left(   \left( \frac{1+\rho|h_m|^2}{\rho|h_m|^2a_n^2+1} \right)  \left( 1+ \rho a^2_n |h_n|^2\right)  > \left( 1+\rho|h_m|^2  \right)^{\frac{1}{2}} \right. \\\nonumber &\times \left. \left( 1+ \rho   |h_n|^2\right)^{\frac{1}{2}}\right)\\\nonumber &=\mathrm{P}\left(    \frac{1+\rho|h_m|^2}{(1+\rho a_n^2|h_m|^2)^2}     >  \frac{   1+ \rho   |h_n|^2 }{ \left( 1+ \rho a^2_n |h_n|^2\right)^2}\right).
\end{align}

After some algebraic manipulations, this probability can be rewritten as follows:
\begin{align}
&\mathrm{P}(R_m+R_n>\bar{R}_m+\bar{R}_n)\\\nonumber &=\mathrm{P}\left( \rho   (|h_m|^2+|h_n|^2)+\rho^2   |h_m|^2|h_n|^2>\frac{1-2a_n^2}{a_n^4}\right).
\end{align}
The right-hand  side of the above inequality is non-negative  since the $n$-th user will get less power than the $m$-th user, i.e., $a_n^2\leq \frac{1}{2}$.
Note that the  joint pdf of $\rho|h_m|^2$ and $\rho|h_n|^2$ is given by \cite{David03}
 \begin{align}\nonumber
&f_{|h_m|^2,|h_n|^2}(x,y) = \varpi_1 f(x) f(y)[F(x)]^{m-1}\left(1-F(y)\right)^{M-n}\\\label{joint pdf} &\times \left(F(y)-F(x)\right)^{n-1-m}.
\end{align}
 In addition, the marginal pdf of $|h_n|^2$ is given by
 \begin{align}\label{pdf1}
f_{|h_n|^2}(y) = \varpi_3 f(y) \left(F(y)\right)^{n-1}\left(1-F(y)\right)^{M-n}.
\end{align}

By applying the above density functions,  the addressed probability can be expressed as follows:
\begin{align}\label{QQQQ}
&\mathrm{P}(R_m+R_n>\bar{R}_m+\bar{R}_n)= \underset{Q_2}{\underbrace{{\int_{\varpi_2}^{\infty} } f_{|h_n|^2}(y) dy}}\\\nonumber &+\underset{Q_1}{\underbrace{\underset{    (x+y)+xy>\varpi_2,x<y<\varpi_2}{\int\int} f_{|h_m|^2,|h_n|^2}(x,y) dx dy}}.
\end{align}
Note that the integral range for $x$ in $Q_1$ is $\frac{\varpi_2-y}{1+y}<x<y$, and this range  implies that $\frac{\varpi_2-y}{1+y}<y$, which causes an additional constraint on $y$, i.e., $y>\sqrt{1+\varpi_2}-1$.
 By applying the binomial expansion,  the joint pdf can be further written  as follows:
\begin{align}\nonumber
&f_{|h_m|^2,|h_n|^2}(x,y) = \varpi_1 \sum^{n-1-m}_{i=0}{n-1-m \choose i} (-1)^i f(x) \\\nonumber &\times f(y)[F(x)]^{m-1+i}   \left(1-F(y)\right)^{M-n}(F(y))^{n-1-m-i}  .
\end{align}
Therefore the probability $Q_1$ can now be evaluated as follows:
{\small \begin{align}\nonumber
Q_1&=\varpi_1 \sum^{n-1-m}_{i=0}{n-1-m \choose i} \frac{(-1)^i}{m+i} \int^{\varpi_2}_{\varpi_4} f(y)(F(y))^{n-1-m-i} \\ \label{Q11} &\times    \left(1-F(y)\right)^{M-n}  \left(\left[F\left(y\right)\right]^{m+i}-\left[F\left(\frac{\varpi_2-y}{1+y}\right)
\right]^{m+i}\right)dy.
\end{align}}

On the other hand, $Q_2$ can be calculated as follows:
\begin{align}
Q_2 &=\int_{\varpi_2}^{\infty} \varpi_3 f(y) \left(F(y)\right)^{n-1}\left(1-F(y)\right)^{M-n} dy\\\nonumber
&=\int_{\varpi_2}^{\infty} \varpi_3 \frac{1}{\rho}e^{-\frac{(M-n+1)y}{\rho}} \left(1- e^{-\frac{y}{\rho}}\right)^{n-1}  dy.
\end{align}
By applying the binomial expansion, $Q_2$ can be written as follows:
\begin{align}\label{Q2}
Q_2&= \frac{\varpi_3}{\rho}\sum^{n-1}_{j=0}{n-1 \choose j}(-1)^j \int_{\varpi_2}^{\infty} e^{-\frac{(M-n+j+1)y}{\rho}}    dy
\\\nonumber &= \frac{\varpi_3}{\rho}\sum^{n-1}_{j=0}{n-1 \choose j}(-1)^j \frac{\rho}{M-n+j+1} e^{-\frac{(M-n+j+1)\varpi_2}{\rho}}.
\end{align}

Combining \eqref{Q11} with \eqref{Q2}, the first part of the theorem is proved.

To find high SNR approximations for $Q_1$ and $Q_2$, first observe that the  integral in \eqref{Q11} is calculated for the range of $0\leq y<\varpi_2$. At high SNR, the two functions $f(y)$ and $F(y)$ can be approximated as follows:  $f(y)=\frac{1}{\rho}e^{-\frac{y}{\rho}}\approx \frac{1}{\rho}$ and $F(y) = 1-e^{-\frac{y}{\rho}}\approx \frac{y}{\rho}$, since $0\leq y\leq \varpi_2$ and $\rho\rightarrow \infty$.

Define $u(y) = \frac{\varpi_2-y}{1+y}$. It is straightforward to show $$0\leq u(y)\leq \varpi_2,$$ for $0\leq y\leq \varpi_2$ ,  since $\frac{d g(y)}{d y}<0$. Therefore at high SNR, we can have the following approximation:
\[
F\left(\frac{\varpi_2-y}{1+y}\right) = 1-e^{-\frac{\varpi_2-y}{\rho(1+y)}}\approx \frac{\varpi_2-y}{\rho(1+y)}.
\]

Now the probability   $Q_1$ can be approximated as follows:
\begin{align}\nonumber
Q_1&\approx\varpi_1 \sum^{n-1-m}_{i=0}{n-1-m \choose i} \frac{(-1)^i}{m+i} \int^{\varpi_2}_{\varpi_4} \frac{1}{\rho} \left(\frac{y}{\rho}\right)^{n-1-m-i} \\\nonumber &\times       \left(\left[\frac{y}{\rho}\right]^{m+i}-\left[\frac{\varpi_2-y}{\rho(1+y)}\right]^{m+i}\right)dy\\\nonumber
&\approx\frac{\varpi_1}{\rho^n} \sum^{n-1-m}_{i=0}{n-1-m \choose i} \frac{(-1)^i}{m+i} \int^{\varpi_2}_{\varpi_4} y^{n-1-m-i} \\ \label{app Q1} &\times       \left( y^{m+i}-\left[\frac{\varpi_2-y}{(1+y)}\right]^{m+i}\right)dy.
\end{align}

The high SNR approximation for $Q_2$ is more complicated. After applying the series expansion of the exponential  functions in \eqref{Q2}, we have
\begin{align}
&Q_2 =\sum^{\infty}_{i=0} \frac{\varpi_3}{\rho}\sum^{n-1}_{j=0}{n-1 \choose j} \frac{(-1)^{i+j} \frac{(M-n+j+1)^{i-1}\varpi_2^i}{\rho^{i-1}}}{i!}\\ \nonumber
&=\sum^{\infty}_{i=0} \frac{(-1)^i\varpi_3\varpi_2^i}{i!\rho^i}\sum^{n-1}_{j=0}{n-1 \choose j}(-1)^j  (M-n+j+1)^{i-1}.
\end{align}

Consider $Q_2$ as a function of $\varpi_2$, and $Q_2=1$  is true  for $\varpi_2=0$, as can be seen from the definition  of $Q_2$ in \eqref{QQQQ}. %But   $\varpi_2=0$ means $a_n^2=\frac{1}{2}$ and $|h_m|^2=|h_n|^2$, whose probability is zero.
On the other hand by letting $\varpi_2=0$ in \eqref{Q2}, we obtain the following equality:
\begin{align}
 {\varpi_3} \sum^{n-1}_{j=0}{n-1 \choose j}(-1)^j \frac{1}{M-n+j+1}  =1.
\end{align}
Consequently $Q_2$ can be rewritten as follows:
\begin{align}\label{Q_23}
Q_2
&=1+\sum^{\infty}_{i=1} \frac{(-1)^i\varpi_3\varpi_2^i}{i!\rho^i}\sum^{n-1}_{j=0}{n-1 \choose j}(-1)^j  \\ \nonumber &\times \sum^{i-1}_{l=0}{i-1 \choose l} (M-n+1)^{i-1-l}j^l.
\end{align}
Recall the following sums of the binomial coefficients (Eq. (0.154.3) in \cite{GRADSHTEYN}):
\begin{align}\label{property1}
\sum^{n-1}_{j=0}{n-1 \choose j}(-1)^j j^l =0,
\end{align}
for $n-2\geq l\geq 1$ and
\begin{align}\label{property2}
\sum^{n-1}_{j=0}{n-1 \choose j}(-1)^j j^{n-1} =(-1)^{n-1}(n-1)!.
\end{align}
Therefore all the components in \eqref{Q_23} containing $j^{l}$, $l< (n-1)$, can be removed, since they are equal to zero by using \eqref{property1}. Furthermore, all the components containing $j^{l}$, $l> (n-1)$ can also be ignored, since the one with $j=n-1$ is the dominant factor. With these steps, the probability can be approximated as follows:
\begin{align}\label{app Q2}
Q_2
&\approx1+  \frac{(-1)^n\varpi_3\varpi_2^n}{n!\rho^n}(-1)^{n-1} (n-1)!\\ \nonumber &=1-  \frac{ \varpi_3\varpi_2^n}{n\rho^n} .
\end{align}
Combining \eqref{app Q1} and \eqref{app Q2}, the second part of the theorem is also proved.
 \hspace{\fill}$\blacksquare$\newline

\textsl{Proof  for Theorem \ref{theorm2}:}

Recall that  the outage performance of the $n$-th user is given by
\begin{align}
&\mathrm{P}\left(\log(1+a_n^2\rho|h_n|^2)<R\right)\\ \nonumber =&\underset{Q_3}{\underbrace{\mathrm{P}\left(\log\left(1+\frac{|h_m|^2 -\frac{I}{\rho}}{|h_m|^2(1+I)}\rho|h_n|^2\right)<R, |h_m|^2 >\frac{I}{\rho}\right) }}\\ \nonumber &+ \underset{Q_4}{\underbrace{\mathrm{P}\left(  |h_m|^2 <\frac{I}{\rho}\right)}}.
\end{align}
The first factor in the above equation can be calculated as follows:
\begin{align}
Q_3=& \mathrm{P}\left( \frac{|h_m|^2 -\frac{I}{\rho}}{|h_m|^2(1+I)} |h_n|^2 <\epsilon_1, |h_m|^2 >\frac{I}{\rho}\right).
\end{align}
  Recall  that the users' channels are ordered, i.e., $|h_m|^2<|h_n|^2$, which brings additional  constraints to the integral range in the above equation.  The constraints can be written  as follows:
\begin{align}
b<|h_m|^2<\min\left\{|h_n|^2, \frac{b}{1-\frac{a\epsilon_1}{|h_n|^2}}\right\}.
\end{align}
  The  outage events due to these constraints  can be classified as follows:
\begin{enumerate}
\item If $|h_n|^2<a\epsilon_1$, we   have the following:
\begin{align}
&\mathrm{P}\left( \frac{|h_m|^2 -\frac{I}{\rho}}{|h_m|^2(1+I)} |h_n|^2 <\epsilon_1\right)\\ \nonumber &=\mathrm{P}\left( |h_m|^2(|h_n|^2 -  \epsilon_1a) <b|h_n|^2\right)=1.
\end{align}
Therefore the probability $Q_3$ can be expressed as follows: \footnote{It is assumed that $b\leq a\epsilon_1$ here. For the case of $b>a\epsilon_1$, the outage probability can be calculated in a straightforward way, since  there will be fewer  events to analyze. Note that  the same diversity order will  be obtained regardless of the choice of $b$ and $a\epsilon_1$.}
$$Q_3=\mathrm{P}\left(b\leq |h_n|^2<a\epsilon_1,|h_n|^2>|h_m|^2>b\right).$$
    \item If $|h_n|^2>a\epsilon_1$, there are two possible events:
\begin{enumerate}
\item If $|h_n|^2>b+a\epsilon_1$, we have $\frac{b}{1-\frac{a\epsilon_1}{|h_n|^2}}<|h_n|^2$, and $Q_3$ can be written as follows:
    \[\hspace{-1.5em}
    Q_3 = \mathrm{P}\left(|h_n|^2>b+a\epsilon_1, b<|h_m|^2<\frac{b}{1-\frac{a\epsilon_1}{|h_n|^2}}\right).
    \]
\item If $|h_n|^2<b+a\epsilon_1$, we have $\frac{b}{1-\frac{a\epsilon_1}{|h_n|^2}}>|h_n|^2$, and $Q_3$ can be written as follows:
    \[\hspace{-1.5em}
    Q_3 = \mathrm{P}\left(a\epsilon_1<|h_n|^2<b+a\epsilon_1, b<|h_m|^2<|h_n|^2\right),
    \] which is again  conditioned on $b<a\epsilon_1$.
\end{enumerate}
\end{enumerate}
Therefore, the probability $Q_3$ can be written as follows:
\begin{align} \label{equ111}
Q_3&= \mathrm{P}\left(b\leq |h_n|^2<a\epsilon_1,|h_n|^2>|h_m|^2>b\right) \\\nonumber &+\mathrm{P}\left(|h_n|^2<b+a\epsilon_1, b<|h_m|^2<|h_n|^2\right)\\ \nonumber  &+ \mathrm{P}\left(|h_n|^2>b+a\epsilon_1, b<|h_m|^2<\frac{b}{1-\frac{a\epsilon_1}{|h_n|^2}}\right).
\end{align}
The first probability in \eqref{equ111} can be   calculated by applying \eqref{pdf1} as  follows:
\begin{align}\nonumber
 & \mathrm{P}\left(b\leq |h_n|^2<a\epsilon_1,|h_n|^2>|h_m|^2>b\right)\\ \nonumber &= \sum^{n-1-m}_{i=0}{n-1-m \choose i} (-1)^i\int^{a\epsilon_1}_{b} g(y) \left(1-G(y)\right)^{M-n}\\\nonumber &\times G(y)^{n-1-m-i}\int^{y}_{b} \varpi_1 g(x) [G(x)]^{m-1+i}dxdy\\ \nonumber &= \sum^{n-1-m}_{i=0}{n-1-m \choose i} (-1)^i\int^{a\epsilon_1}_b g(y) \left(1-G(y)\right)^{M-n}\\\nonumber &\times G(y)^{n-1-m-i} \varpi_1  \frac{\left(G(y)^{m+i} - G(b)^{m+i}\right)}{m+i}dy.
\end{align}
Following similar steps, the second probability in \eqref{equ111} can be expressed as
\begin{align}
&\mathrm{P}\left(a\epsilon_1<|h_n|^2<b+a\epsilon_1, b<|h_m|^2<|h_n|^2\right) \\ \nonumber &= \sum^{n-1-m}_{i=0}{n-1-m \choose i} (-1)^i\int^{b+a\epsilon_1}_{a\epsilon_1} g(y) \left(1-G(y)\right)^{M-n}\\\nonumber &\times G(y)^{n-1-m-i} \varpi_1  \frac{\left(G(y)^{m+i} - G(b)^{m+i}\right)}{m+i}dy.
\end{align}

The third probability in \eqref{equ111} can be calculated  as follows:
\begin{align}\label{eq3}
& \mathrm{P}\left(|h_n|^2>b+a\epsilon_1, b<|h_m|^2<\frac{b}{1-\frac{a\epsilon_1}{|h_n|^2}}\right)\\ \nonumber  &= \sum^{n-1-m}_{i=0}{n-1-m \choose i} (-1)^i\int^\infty_{b+a\epsilon_1} g(y) \left(1-G(y)\right)^{M-n}\\ \nonumber &\times G(y)^{n-1-m-i}\int^{\frac{b}{1-\frac{a\epsilon_1}{|h_n|^2}}}_{b} \varpi_1 g(x) [G(x)]^{m-1+i}dxdy\\ \nonumber &= \sum^{n-1-m}_{i=0}{n-1-m \choose i} (-1)^i\int^\infty_{b+a\epsilon_1} g(y) \left(1-G(y)\right)^{M-n}\\\nonumber &\times G(y)^{n-1-m-i} \varpi_1  \frac{\left(G\left(\frac{b}{1-\frac{a\epsilon_1}{|h_n|^2}}\right)^{m+i} - G(b)^{m+i}\right)}{m+i}dy.
\end{align}
Note that $Q_4$ can be obtained easily by applying \eqref{pdf1} and the first part of the theorem is proved.

Recall that   the first probability in \eqref{equ111} can be expressed  as follows:
\begin{align}\nonumber
 & \mathrm{P}\left(b\leq |h_n|^2<a\epsilon_1,|h_n|^2>|h_m|^2>b\right) \\ \nonumber &= \varpi_1 \sum^{n-1-m}_{i=0}{n-1-m \choose i} (-1)^i\int^{a\epsilon_1}_b g(y) \\\nonumber &\times \left(1-G(y)\right)^{M-n}G(y)^{n-1-m-i} \frac{\left(G(y)^{m+i} - G(b)^{m+i}\right)}{m+i}dy,
\end{align}
where the integral range is $0\leq y\leq (a\epsilon_1)$. Note that when $\rho\rightarrow \infty$,   $\epsilon_1$ approaches    zero, which means $y\rightarrow 0$, $g(y)\approx 1$ and   $G(y)\approx 1-y $. Therefore the above probability can be approximated as follows:
\begin{align}\label{eqx3}
 & \mathrm{P}\left(b\leq |h_n|^2<a\epsilon_1,|h_n|^2>|h_m|^2>b\right) \\ \nonumber &\approx \varpi_1 \sum^{n-1-m}_{i=0}{n-1-m \choose i} (-1)^i\\ \nonumber &\times \int^{a\epsilon_1}_b      y^{n-1-m-i} \frac{\left(y^{m+i} -b^{m+i}\right)}{m+i}dy
 \\ \nonumber &\approx \varpi_1 \sum^{n-1-m}_{i=0}{n-1-m \choose i} (-1)^i\\ \nonumber &\times     \frac{\left(\frac{(a\epsilon_1)^{n}-b^n}{m+i+1} -b^{m+i}\left( \frac{(a\epsilon_1)^{n-m-i} - b^{n-m-i}}{n-m-i}\right)\right)}{m+i}\rightarrow \rho^{-n}. \end{align}

Following similar steps, the second probability in \eqref{equ111} can be approximated as follows:
\begin{align}\label{eqx2}
&\mathrm{P}\left(a\epsilon_1<|h_n|^2<b+a\epsilon_1, b<|h_m|^2<|h_n|^2\right)\rightarrow \rho^{-n}.
\end{align}

The exact diversity order of the third probability in \eqref{equ111} is difficult to obtain. Particularly the expression in \eqref{eq3} is   difficult to use  for asymptotic studies, since the range of $y$ is not limited and those manipulations related to high SNR approximations cannot be applied here. We first rewrite  \eqref{eq3}  in an alternative form as follows:
\begin{align}\nonumber
& \mathrm{P}\left(|h_n|^2>b+a\epsilon_1, b<|h_m|^2<\frac{b}{1-\frac{a\epsilon_1}{|h_n|^2}}\right)\\ \nonumber  &= \mathrm{P}\left( b<|h_m|^2<b+a\epsilon_1, b+a\epsilon_1<|h_n|^2<\frac{|h_m|^2a\epsilon_1}{|h_m|^2-b}\right).
\end{align}
Note that $b+a\epsilon_1< \frac{|h_m|^2a\epsilon_1}{|h_m|^2-b}$ always holds since $|h_m|^2<\frac{b}{1-\frac{a\epsilon_1}{|h_n|^2}}$.

Now applying the joint pdf of the two channel coefficients, we obtain the following expression:
\begin{align}\nonumber
 & \mathrm{P}\left( b<|h_m|^2<b+a\epsilon_1, b+a\epsilon_1<|h_n|^2<\frac{|h_m|^2a\epsilon_1}{|h_m|^2-b}\right)\\\nonumber &= \sum^{n-1-m}_{i=0}\varpi_1{n-1-m \choose i}(-1)^i\int^{b+a\epsilon_1}_{b} g(x) [G(x)]^{m-1+i} \\\nonumber &\times \int^{G\left(\frac{xa\epsilon_1}{x-b}\right)}_{G(b+a\epsilon_1)}  \left[G(y) \right]^{n-1-m-i} \left(1-G(y)\right)^{M-n}dG(y)dx.
\end{align}
Again applying the binomial expansion, the above probability  can be further expanded as follows:
\begin{align}\nonumber
 & \mathrm{P}\left( b<|h_m|^2<b+a\epsilon_1, b+a\epsilon_1<|h_n|^2<\frac{|h_m|^2a\epsilon_1}{|h_m|^2-b}\right)\\\nonumber &= \sum^{n-1-m}_{i=0}\varpi_1{n-1-m \choose i}(-1)^i\int^{b+a\epsilon_1}_{b} g(x) [G(x)]^{m-1+i} \\\nonumber &\times \sum^{M-n}_{j=0}{M-n \choose j}(-1)^j \int^{G\left(\frac{xa\epsilon_1}{x-b}\right)}_{G(b+a\epsilon_1)}  \left[G(y) \right]^{n-1-m-i+j} dG(y)dx
 \\\nonumber &= \sum^{n-1-m}_{i=0}\varpi_1{n-1-m \choose i}(-1)^i\int^{b+a\epsilon_1}_{b} g(x) [G(x)]^{m-1+i} \\\label{eq3-1} &\times \sum^{M-n}_{j=0}{M-n \choose j}\frac{(-1)^j}{n-1-m-i+j} \\ \nonumber &\times  \left(\left[G\left(\frac{xa\epsilon_1}{x-b}\right)\right] ^{n-m-i+j} - \left[G(b+a\epsilon_1)\right]^{n-m-i+j}\right)  dx.
\end{align}
Compared to \eqref{eq3}, the above equation is more complicated; however, this expression is more suitable for asymptotic studies, as explained in the following.

Recall that the integral range in \eqref{eq3-1} is $b<x< b+a\epsilon_1$. When $\rho\rightarrow 0$, we have  $b\rightarrow 0$ and $b+a\epsilon_1\rightarrow 0$, which implies $x\rightarrow 0$. Therefore the following approximation can be obtained:
\begin{align}
 & \mathrm{P}\left( b<|h_m|^2<b+a\epsilon_1, b+a\epsilon_1<|h_n|^2<\frac{|h_m|^2a\epsilon_1}{|h_m|^2-b}\right)\\\nonumber  &\approx \sum^{n-1-m}_{i=0}\varpi_1{n-1-m \choose i}(-1)^i \\\nonumber &\times \sum^{M-n}_{j=0}{M-n \choose j}\frac{(-1)^j}{n-1-m-i+j}\int^{b+a\epsilon_1}_{b}   x^{m-1+i} \\ \nonumber &\times  \left(\left[G\left(\frac{xa\epsilon_1}{x-b}\right)\right] ^{n-m-i+j} - \left[b+a\epsilon_1\right]^{n-m-i+j}\right)  dx.
\end{align}

First focus on the following integral which is from the above equation:
\begin{align} &\int^{b+a\epsilon_1}_{b}   x^{m-1+i}    \left[G\left(\frac{xa\epsilon_1}{x-b}\right)\right] ^{n-m-i+j}    dx
\\\nonumber &\approx \int^{a\epsilon_1}_{0}   (b+z)^{m-1+i}    \left[1- e^{- \frac{ab\epsilon_1}{z} }\right] ^{n-m-i+j}    dz.
\end{align}
We can find   the following bounds for the above integral:
\begin{align}\label{bounds}
 &  \int^{b+a\epsilon_1}_{b}   x^{m-1+i}       dx\\\nonumber &\geq\int^{b+a\epsilon_1}_{b}   x^{m-1+i}    \left[G\left(\frac{xa\epsilon_1}{x-b}\right)\right] ^{n-m-i+j}    dx \\\nonumber &\geq \int^{a\epsilon_1}_{0}   (b+z)^{m-1+i}    \left[1- \frac{1}{1 +  \frac{ab\epsilon_1}{z} }\right] ^{n-m-i+j}    dz,
\end{align}
where the lower bound is obtained due to the inequality  $$e^{- \frac{ab\epsilon_1}{z} }\leq \frac{1}{1 +  \frac{ab\epsilon_1}{z} },$$ when $0\leq z\leq a\epsilon_1$.
The upper bound in \eqref{bounds} can be approximated at high SNR as follows:
\begin{align}\label{uper bounds}
 &  \int^{b+a\epsilon_1}_{b}   x^{m-1+i}       dx\\\nonumber &=\frac{(b+a\epsilon_1)^{m+i} - b^{m+i}}{m+i}\rightarrow \frac{1}{\rho^{m+i}}.
\end{align}
On the other hand, the lower bound in \eqref{bounds} can be approximated as follows:
 \begin{align}
  &  \int^{a\epsilon_1}_{0}   (b+z)^{m-1+i}    \left[1- \frac{1}{1 +  \frac{ab\epsilon_1}{z} }\right] ^{n-m-i+j}    dz \\ \nonumber  &=(ab\epsilon_1)^{n-m-i+j}\int^{a\epsilon_1+ab\epsilon_1}_{ab\epsilon_1}   \frac{(w+b-ab\epsilon_1)^{m-1+i} }{w^{n-m-i+j}  }     dz
      \\ \nonumber &=(ab\epsilon_1)^{n-m-i+j}\sum^{m-1+i}_{k=0}(b-ab\epsilon_1)^{m-1+i-k}
     \\ \nonumber &\times  \int^{a\epsilon_1+ab\epsilon_1}_{ab\epsilon_1}    w^{k-(n-m-i+j)} dz\triangleq \sum^{m-1+i}_{k=0}\xi_k.
\end{align}
At high SNR, we can show that
 \begin{align}
 \xi_k\rightarrow \left\{ \begin{array}{ll} \rho^{-(n+j)}\ln \rho,  &\text{for}\quad k+1=n-m-i+j \\ \rho^{-(n+j)},  &\text{otherwise}   \end{array}\right.\hspace{-1em}.
\end{align}
Since $\frac{\log \log \rho}{\log \rho}\rightarrow 0$ for $\rho\rightarrow \infty$,
the lower bound in \eqref{bounds} can be approximated as follows:
{\small \begin{align}\label{lower bound}
  \int^{a\epsilon_1}_{0}   (b+z)^{m-1+i}    \left[ \frac{ab\epsilon_1}{z +  ab\epsilon_1 }\right] ^{n-m-i+j}    dz\rightarrow \rho^{-(n+j)}.
\end{align}}
Based on the upper and lower bounds in \eqref{uper bounds} and \eqref{lower bound} and after  some algebraic manipulation, we have the following inequality:
{\small \begin{align}\label{eqx1}
\rho^{-n}\dot\leq\mathrm{P}\left(|h_n|^2>b+a\epsilon_1, b<|h_m|^2<\frac{b}{1-\frac{a\epsilon_1}{|h_n|^2}}\right)\dot\leq \rho^{-m},
\end{align}}
where $a\dot\leq b$ denotes $\left(-\frac{\log a}{\log\rho}\right)\leq\left(-\frac{\log b}{\log\rho}\right)$ when $\rho\rightarrow \infty$ \cite{Zhengl03}.

Combining \eqref{eqx3}, \eqref{eqx2} and \eqref{eqx1}, we can obtain the following asymptotic bounds:
\begin{align}
\rho^{-n}\dot\leq Q_2\dot\leq \rho^{-m}.
\end{align}
Following   similar steps as above, we can also find that $Q_3\doteq \rho^{-m}$, which is dominant in $\mathrm{P}_o^n$, and the proof for the second part of the theorem is   completed.
\hspace{\fill}$\blacksquare$\newline
  \end{document}